\documentclass[11pt]{article}
\usepackage[english]{babel}
\usepackage{amsmath,amssymb,amsthm}
\usepackage[pdfpagemode=UseOutlines,bookmarks=true,
bookmarksopen=true,bookmarksnumbered=true,unicode=true,
pdffitwindow=true,pdfpagelayout=SinglePage,pdfhighlight=/P,
colorlinks=true,linkcolor=blue,citecolor=blue,urlcolor=blue]{hyperref}

\theoremstyle{plain}
\newtheorem{theorem}{Theorem}
\newtheorem{statement}[theorem]{Statement}
\newtheorem{corollary}[theorem]{Corollary}
\newtheorem{lemma}[theorem]{Lemma}

\theoremstyle{remark}
\newtheorem{example}{Example}
\newtheorem{remark}{Remark}

\newcommand{\Complex}{\mathbb{C}}
\newcommand{\Integer}{\mathbb{Z}}
\newcommand{\coef}{\mathop{\rm coef}\nolimits}
\newcommand{\const}{\mathop{\rm const}}
\newcommand{\res}{\mathop{\rm res}}
\newcommand{\Res}{\mathop{\rm Res}}
\newcommand{\im}{\mathop{\rm Im}}
\newcommand{\Eu}{{\sf E}}
\newcommand{\FF}{{\cal F}}
\newcommand{\nFF}{\widetilde\FF}

\title{Necessary integrability conditions for evolutionary lattice equations}
\author{V.E. Adler\thanks{L.D. Landau Institute for theoretical physics,
Ak. Semenov str. 1-A, 142432 Chernogolovka, Russian Federation. E-mail: {\tt
adler@itp.ac.ru}}}

\date{May 29, 2014}

\begin{document}\thispagestyle{empty}
\maketitle
\begin{abstract}
The structure of solutions is studied for the Lax equation $D_t(G)=[F,G]$ for
formal power series with respect to the shift operator. It is proved that if
the equation with a given series $F$ of degree $m$ admits a solution $G$ of
degree $k$ then it admits, as well, a solution $H$ of degree $m$ such that
$H^k=G^m$. This property is used for derivation of necessary integrability
conditions for scalar evolutionary lattices.\medskip

\noindent{\em Keywords:}\/  Volterra type lattice, higher symmetry,
conservation law, integrability test
\smallskip

\noindent MSC 37K10\qquad  PACS 02.30.Ik
\end{abstract}

%-------------------------------------------------------------------------------
\section{Introduction}\label{s:i}

It is well known that existence of infinite sets of higher symmetries and
conservation laws is a characteristic property of integrable equations. In the
case of two-dimensional evolutionary equations $\partial_t(u)=f[u]$, this
implies existence of formal operator series $G,R$ which satisfy the equations
(see notations in section \ref{s:G})
\begin{equation}\label{GR}
 D_t(G)=[f_*,G],\quad D_t(R)+f_*^\dag R+Rf_*=0.
\end{equation}
The solvability of (\ref{GR}) with respect to the coefficients of $G,R$
provides the necessary conditions of integrability. This approach has been
applied, in the papers by Shabat et al, for classification of integrable
equations, both in partial derivatives \cite{Sokolov_Shabat_1984,
Mikhailov_Shabat_Yamilov_1987, Sokolov_1988, Mikhailov_Shabat_Sokolov_1991,
Mikhailov_Shabat_1993, Meshkov_Sokolov_2013} and differential-difference ones
\cite{Yamilov_1983, Mikhailov_Shabat_Yamilov_1987, Shabat_Yamilov_1991,
Levi_Yamilov_1997, Yamilov_2006, Adler_2008}. The continuous and discrete
equations have much in common, but there are differences as well.

In the continuous case, $G$ and $R$ are pseudodifferential operators, that is,
Laurent series with respect to inverse powers of total derivative $D$. If we
consider 1-component evolutionary equations then the coefficients of the series
are scalar functions of dynamical variables $u$. The root extraction plays a
very important role in the theory of such series. This operation is defined for
a generic series $G=g_kD^k+g_{k-1}D^{k-1}+\dots$, since the coefficients of the
series $G^{1/k}$ are computed by explicit algebraic formulae which do not lead
out from the coefficient field in consideration \cite{Schur_1905}. This
property drastically simplifies the study of equations (\ref{GR}), because it
allows us to set $\deg G=1$, $\deg R=0$ without loss of generality. In
particular, the integrability conditions can be rewritten in the form of a
sequence of conservation laws
\begin{equation}\label{Drs}
 D_t(\rho_j)=D(\sigma_j),\quad j=0,1,\dots
\end{equation}
where $\rho_j$ are explicitly expressed through $\rho_i,\sigma_i$ at $i<j$. If
the left hand side belongs to the image of $D$ (which coincides with the kernel
of the variational derivative) then one can find $\sigma_j$ and pass to the
next step of integrability test.

The situation is more complicated in the case of the lattice equations
\begin{equation}\label{utn}
 \partial_t(u_n)=f(u_{n+m},\dots,u_{n-m}),\quad n\in\Integer
\end{equation}
studied in the present paper. Here, $G,R$ are power series with respect to the
shift operator $T$ and the root $G^{1/k}$ does not exist for a generic series
$G=g_kT^k+g_{k-1}T^{k-1}+\cdots$. It is clear already from consideration of the
leading coefficient which must be of a special form $g_k=hT(h)\cdots
T^{k-1}(h)$ in order that the root exists. As a result, the integrability
conditions are more involved and cannot be cast into the form of conservation
laws, in general; equations (\ref{GR}) lose their effectiveness, because the
degrees of the series $G,R$ are not known in advance. However, it turns out
that, in fact, the integrability conditions do not depend on $k$. The goal of
this paper is to prove that the general case can be reduced to $\deg G=m$,
$0\le\deg R<m$.

In section \ref{s:H}, the following statement is proved: if a difference Lax
equation $D_t(G)=[F,G]$ with $\deg F=m$ admits a solution $G$ of degree $k$,
then there exists its another solution $H=G^{m/k}$ of degree $m$. The key idea
is that the coefficients of $H$ can be computed explicitly by use of equations
$D_t(H)=[F,H]$ and $H^k=G^m$ simultaneously, after this it is possible to prove
that each equation is fulfilled separately. In other words, extraction of the
root is still possible in a certain weak sense---due to the fact that the
solutions of the Lax equation are far from being generic series and their
coefficients already possess a special structure.

This observation is used in sections \ref{s:G}, \ref{s:R} in order to formulate
the necessary integrability conditions for the lattices (\ref{utn}). If $m>1$
then these conditions remain more complicated comparing to the continuous case,
but their form is quite suitable for testing of a given equation. Solving of
the Lax equation with respect to the coefficients of $G$ amounts to the
checking whether a given expression belongs to the image of an operator of the
form $T^m-T^j(a)/a$ where $a$ is a fixed function, and to the computation of
its preimage. In principle, both problems admit algorithmic solutions which
are, however, beyond the scope of this paper. The analysis of conditions in a
general form, in order to obtain classification results or to construct novel
examples is, even in the case $m=2$, a very difficult task which requires a
separate study as well. It should be noted that all examples with $m>1$ known
at the moment are equivalent (up to Miura type substitutions) to the
Bogoyavlensky lattices \cite{Bogoyavlensky_1991} and some their generalizations
\cite{Adler_Postnikov_2011, Garifullin_Yamilov_2012,
Garifullin_Mikhailov_Yamilov_2014}. In this respect, the theory lags behind the
continuous case where a number of classification results was obtained for the
Burgers type equations of orders 2,4 \cite{Svinolupov_1985} and the KdV type
equations of orders 3,5,7 (see references in \cite{Meshkov_Sokolov_2013}).

The integrable equations (\ref{utn}) at $m=1$ (the Volterra type lattices) were
classified by Yamilov \cite{Yamilov_1983}. In this case, the necessary
integrability conditions are of the form analogous to (\ref{Drs}):
\begin{equation}\label{Trs}
 D_t(\rho_j)=(T-1)(\sigma_j),\quad
 \rho_j-\bar\rho_j=(T-1)(s_j),\quad j=0,1,\dots
\end{equation}
where $\rho_j,\bar\rho_j$ are expressed through
$\rho_i,\bar\rho_i,\sigma_i,s_i$ at $i<j$. The derivation of conditions
(\ref{Trs}) in papers \cite{Levi_Yamilov_1997, Yamilov_2006} was based on the
assumption that the lattice (\ref{utn}) admits higher symmetries of orders $k$,
$k+1$ where $k$ is arbitrary large, which implies the existence of a series $G$
of degree 1. On the other hand, the authors noted that it were possible to
derive the same conditions by use of the series $G$ of any degree, although by
means of more involved computations. This is completely explained by the
procedure of root extraction described above, moreover, the assumption on the
symmetry of order $k+1$ becomes redundant.

Regarding the method of derivation of concrete expressions for the conditions
like (\ref{Drs}) or (\ref{Trs}), let us recall that the densities $\rho_j$ for
the continuous Lax equations can be computed in two ways: as the residues of
the fractional powers $\res G^{j/k}$ (Gelfand, Dikii
\cite{Gelfand_Dikii_1976}), or as the coefficients of an expansion with respect
to $\lambda$ for the logarithmic derivative of the formal Baker--Akhiezer
function (Wilson \cite{Wilson_1981}, the idea goes up to the construction of
the generating function for the conservation laws by inversion of the Miura
transformation \cite{Miura_Gardner_Kruskal_1968}). The equivalence of both
definitions was established by Wilson \cite{Wilson_1981} and Flaschka
\cite{Flaschka_1983}. A detailed description of the method based on the
residues in the context of derivation of the necessary integrability conditions
can be found, e.g. in \cite{Sokolov_Shabat_1984,
Mikhailov_Shabat_Yamilov_1987}. This method requires more involved
computations comparing to the method based on the expansion of the formal
$\psi$-function \cite{Chen_Lee_Liu_1979, Meshkov_1994, Meshkov_Sokolov_2013}.
Both methods work in the difference setting as well \cite{Kupershmidt_1985},
but only under assumption that $\deg G=1$, which is an essential stipulation,
as we have seen.

In section \ref{s:1}, the conditions (\ref{Trs}) are derived by use of
expansion of the formal $\psi$-function which is much simpler than computing
$\Res G^j$ and allows us to obtain explicit closed expressions for the
densities in terms of the Bell polynomials. Though, the extraction of the list
of integrable lattices at $m=1$ requires, according to the Yamilov's results
\cite{Yamilov_1983, Yamilov_2006}, just three simplest conditions which can be
derived without any theory under very modest assumptions about symmetries and
conservations laws. This offers hope that in the case, say, $m=2$, the
classification requires not too many integrability conditions as well.

%-------------------------------------------------------------------------------
\section{Extraction of the root in the difference setting}\label{s:H}

Let $\FF$ be the field of locally analytical functions of finite number of
dynamical variables $u_n$, $n\in\Integer$ and let the rule
\[
 T(a(u_i,\dots,u_j))=a(u_{i+1},\dots,u_{j+1})
\]
define the action of the shift operator $T$ on functions from $\FF$. The rule
$aT^ibT^j=aT^i(b)T^{i+j}$, being distributed over addition, defines a
multiplication of the difference operators. The formal Laurent series with
respect to the negative or positive powers of $T$ constitute the division rings
\[
 \FF((T^{-1}))=\{\sum_{j<+\infty}a_jT^j\mid a_j\in\FF\},\quad
 \FF((T))=\{\sum_{j>-\infty}a_jT^j\mid a_j\in\FF\}.
\]
All statements in this section are given for the series from $\FF((T^{-1}))$,
the passage to $\FF((T))$ amounts to renaming $u_n\to u_{-n}$, $T\to T^{-1}$.

Let $F=f_mT^m+\ldots\in\FF((T^{-1}))$ be a given series of degree $m>0$ and
$D_t:\FF\to\FF$ be a given evolutionary differentiation (that is, commuting
with $T$; this guarantees that $D_t$ is a differentiation in $\FF((T^{-1}))$ as
well). We are interested in solutions of the Lax equation
\begin{equation}\label{Gt}
 D_t(G)=[F,G]
\end{equation}
as the series $G\in\FF((T^{-1}))$ of degree $k>0$. In contrast to the
continuous situation, the root $G^{1/k}$ is not defined for a generic series
$G$ and the study of solutions (or the obstacles for their existence) cannot be
reduced to the case $k=1$. Nevertheless, it turns out that if a solution $G$
exists then its coefficients possess a certain special structure, such that the
following properties are fulfilled:

(i) equation (\ref{Gt}) admits another solution $H$ of degree $m$, such that
$H^k=G^m$. Here, $m$ may be not a minimal positive power of solutions. Thus, in
the discrete setting the root extraction is possible in a weakened sense. This
property is proved in theorem \ref{th:root1};

(ii) it is possible to choose a solution $G$ among all solutions of degree $m$
such that $G=F+o(T)$ (that is, $\deg(G-F)<1$). It seems obvious at first sight,
because (\ref{Gt}) yields the same subset of equations for the partial sum
$G_{>0}=g_mT^m+\dots+g_1T$ as the equation $[F,G]=0$. However, the general
solution for this subset can contain additional constant parameters comparing
to $F_{>0}$, and the further equations may turn solvable only under certain
choice of these constants. In principle, it may turn out that the whole set of
equations for the coefficients of $G$ is solvable in $\FF$ only for such
constants that $G_{>0}\ne F_{>0}$. The fact that actually this is not the case
is proved in the corollary \ref{cor:root2}.
\smallskip

Before we go on to the proofs, let us consider several concrete equations for
the solution coefficients of the Lax equation, in two simplest examples.

\begin{example}[$m=1$, $k=2$]
Let $F=f_1T+f_0+\dots$ and let equation (\ref{Gt}) admit a solution
$G=g_2T^2+g_1T+\dots$ then is it possible to find a solution $H$ such that
$H^2=G$? Let us consider several first equations for the coefficients of $G$:
\begin{align*}
 0&= f_1T(g_2)-T^2(f_1)g_2,\\
 D_t(g_2)&= f_1T(g_1)-T(f_1)g_1+f_0g_2-T^2(f_0)g_2,\\
 D_t(g_1)&= f_1T(g_0)-f_1g_0+f_0g_1-T(f_0)g_1
   +f_{-1}T^{-1}(g_2)-T^2(f_{-1})g_2.
\end{align*}
The first equation implies $g_2=f_1T(f_1)$ (up to a constant factor), then the
second one is brought to the form
\[
 (T+1)(D_t(\log f_1))=(T-1)(g_1/f_1)-(T^2-1)(f_0)
\]
and this implies that a function $h_0\in\FF$ exists such that
\[
 g_1/f_1=(T+1)(h_0),\quad D_t(\log f_1)=(T-1)(h_0-f_0).
\]
(Here, the properties of the difference operators with constant coefficients
are used: $\ker(T+1)=0$, $\ker(T-1)=\Complex$). Now, an easy computation brings
the third equation to the form
\[
 (T+1)(D_t(h_0))=(T-1)(g_0-h_0^2)-(T^2-1)(f_{-1}T^{-1}(f_1))
\]
which implies that a function $h_{-1}\in\FF$ exists such that
\[
 g_0=h_0^2+(T+1)(h_{-1}T^{-1}(f_1)),\quad
 D_t(h_0)=(T-1)((h_{-1}-f_{-1})T^{-1}(f_1)).
\]
Collecting all together we obtain
\begin{gather*}
 G=f_1T(f_1)T^2+f_1(T+1)(h_0)T
  +h_0^2+(T+1)(h_{-1}T^{-1}(f_1))+o(1)\\
 = (f_1T+h_0+h_{-1}T^{-1})^2+o(1)
\end{gather*}
in support of the conjecture that the root can be extracted indeed.
\end{example}

\begin{example}[$m=2$, $k=2$]
Let $F=f_2T^2+f_1T+\dots$ and equation (\ref{Gt}) possesses a solution
$G=f_2T^2+g_1T+\dots$ then is it possible to choose $g_1=f_1$? Assume $g_1\ne
f_1$, then $F-G$ is a series of degree 1 and we obtain the following relations
by repeating computations from the previous example for the equation
$D_t(G)=[F-G,G]$:
\[
 G=(h_1T+h_0+h_{-1}T^{-1})^2+o(1),\quad f_1-g_1=ch_1,\quad
 f_2=g_2=h_1T(h_1).
\]
Therefore, if the relation $G=H^2$ is proved then it is possible, indeed, to
obtain the solution of the form $G+cH=f_2T^2+f_1T+\dots$ as required.
\end{example}
\smallskip

A demerit of the above computations is that the coefficients of the series
$G^{1/2}$ are found implicitly, by inversion of the operators $T+1$ or $T-1$.
However, we can obtain them also in an explicit form---if it is known in
advance that the desired series does exist. The idea is that in order to
extract the root one should use both equations $D_t(H)=[F,H]$ and $H^2=G$
simultaneously. This brings to the recurrent relations of the form
\[
 f_1T(h_j)-T^j(f_1)h_j=D_t(h_{j+1})+\dots,\quad
 f_1T(h_j)+T^j(f_1)h_j=g_{j+1}+\dots
\]
where the right hand sides contain the coefficients $h_1,h_0,\dots,h_{j+1}$
found on the previous steps. Subtracting yields an explicit expression for
$h_j$. After this one has to verify that each equation is fulfilled indeed. It
is not obvious, but plausible, because the equations are consistent, in the
sense that they imply the equation $D_t(G)=[F,G]$ which is true by assumption.

In order to justify these heuristic arguments, in the proof of theorem
\ref{th:root1}, we will make use of the series with nonautonomous coefficients.
Let us consider an extension of the field $\FF$ given by the ring $\nFF$ with
elements represented by {\em sequences} of locally analytical functions
$a(n)=a(n;u_{r_n},\dots,u_{s_n})$, $n\in\Integer$ where each function in the
sequence depends on its own finite set of dynamical variables. Elements from
$\FF$ are identified with sequences of special type
$a(n)=a(u_{n+r},\dots,u_{n+s})$. By definition, multiplication of the sequences
is done termwise and the operator $T$ acts just by the shift of $n$, that is,
$T^k(a(n))=a(n+k)$. The multiplication in the ring
\[
 \nFF((T^{-1}))=\{\sum_{j<+\infty}a_j(n)T^j \mid
 a_j(n)=a_j(n;u_{r_{j,n}},\dots,u_{s_{j,n}})\in\nFF\}
\]
is defined, as before, by the rule $a(n)T^ib(n)T^j=a(n)b(n+i)T^{i+j}$. One can
easily see that the Lax equations are always solvable in such an extension.

\begin{lemma}\label{l:n-Lax}
Let $F(n)=f_m(n)T^m+\ldots\in\nFF((T^{-1}))$, $m\ge1$ and $f_m(n)\not\equiv0$
for all $n$. Then there exists a unique series
$G(n)=g_k(n)T^k+\ldots\in\nFF((T^{-1}))$, for any degree $k$, which satisfies
the equation (\ref{Gt}) and the prescribed initial conditions
$G(0),G(1),\dots,G(m-1)$.
\end{lemma}
\begin{proof}
The coefficients of $g_j(n)$ are computed step by step at $j=k,k-1,\dots$\,.
The equation for $g_j(n)$ appears from (\ref{Gt}) in the order of $T^{j+m}$ and
it is a recurrent relation of the form
\[
 f_m(n)g_j(n+m)-f_m(n+j)g_j(n)=\dots
\]
where the right hand side contains the members of already found sequences
$g_i(n)$ with $i>j$. From here, all values $g_j(n)$, $n\in\Integer$ are defined
uniquely if the values $g_j(0),\dots,g_j(m-1)$ are given.
\end{proof}

\begin{theorem}\label{th:root1}
Let series $G,F\in\FF((T^{-1}))$ satisfy the Lax equation $D_t(G)=[F,G]$ and
$\deg F=m\ge1$, $\deg G=k\ge1$. Then a series $H\in\FF((T^{-1}))$ of degree $m$
exists, unique up to a factor $1^{1/k}$, such that $D_t(H)=[F,H]$ and
$H^k=G^m$.
\end{theorem}
\begin{proof}
The leading coefficient of the series $H$ satisfies the equation
\[
 f_mT^m(h_m)=T^m(f_m)h_m
\]
with the general solution in $\FF$ of the form $h_m=\const f_m$. The constant
is determined, up to the root of 1, from comparing the leading terms in the
equality $H^k=G^m$.

Let us construct the rest coefficients as sequences $h_j(n)$ from the ring
$\nFF$, according to lemma \ref{l:n-Lax}. We will prove, by induction on $j$,
that there exist unique initial data $H(0),H(1),\dots,H(m-1)$ such that the
conditions
\begin{equation}\label{hkgm}
 (H^k-G^m)|_{n=0}=\dots=(H^k-G^m)|_{n=m-1}=0
\end{equation}
are satisfied. Assume that we have already found the coefficients
$h_m(n),\dots$, $h_{j+1}(n)$ such that the equation $D_t(H)=[F,H]$ is satisfied
up to the order of $T^{m+j+1}$, and equations (\ref{hkgm}) are fulfilled up to
the order of $T^{m(k-1)+j+1}$. Equations for $h_j(n)$ which appear in the next
orders can be written down by use of the operators $A(n)=f_m(n)T^m$ and
$X(n)=h_j(n)T^j$ as follows:
\begin{gather}
\label{AX}
 [A(n),X(n)]=a(n)T^{m+j},\quad n\in\Integer,\\
\nonumber
 A(n)^{k-1}X(n)+A(n)^{k-2}X(n)A(n)+\dots+X(n)A(n)^{k-1}\qquad\\
\label{AAX}
 \qquad\qquad =b(n)T^{(k-1)m+j},\quad n=0,\dots,m-1
\end{gather}
where $a(n),b(n)$ are certain polynomials with respect to the coefficients of
the series $F,G$ and coefficients $h_m,\dots,h_{j+1}$ already found
(differentiated with respect to $t$ among them). One can easily see that
equation (\ref{AAX}) is reduced by use of (\ref{AX}) to equations of the form
c$kX(n)A(n)^{k-1}=c(n)T^{(k-1)m+j}$, that is,
\[
 kh_j(n)f_m(n+j)\cdots f_m(n+j+(k-2)m)=c(n),\quad n=0,\dots,m-1.
\]
This uniquely defines the initial data for the sequence $h_j(n)$ and completes
the induction step.

The series $H^k-G^m$ is, for the constructed solution, a solution of the Lax
equation with vanishing initial data and according to the lemma \ref{l:n-Lax}
it is identically zero. Therefore, equations (\ref{AAX}) are fulfilled for all
$n\in\Integer$, not only for $n$ stated above.

Next, let us notice that the polynomials $a(n),b(n)$ and, therefore, $c(n)$,
are of the same form for all $n$, as functions of the coefficients of the
series $F,G,H$. This means that if $h_j(n)=p[F,G,H_{>j}]$, where $p$ is a
function of finite number of variables, then $h_j(n+1)=p[T(F),T(G),T(H_{>j})]$.
Since all $f_i,g_i\in\FF$ and $h_m=f_m\in\FF$, hence we prove, again by
induction, that all $h_j\in\FF$.
\end{proof}

\begin{remark}\label{rem:approx}
It is clear from the proof that a computation of the first $r$ coefficients of
$H$ requires exactly $r$ coefficients of $G$ (not taking the coefficients of
$F$ into account). Indeed, the series $G$ is used only in the initial
conditions (\ref{hkgm}) where the number of coefficients of both series
coincides in each order of $T$.
\end{remark}

It is easy to demonstrate that if the Lax equation (\ref{Gt}) possesses at
least one nontrivial solution (that is, different from $cT^0$) then its general
solution is represented as a series with constant coefficients
\[
 G=\sum_{j<+\infty}c_jH^j
\]
where $H$ is a solution of minimal positive degree $d$. This implies, in
particular, that any two solutions commute and it follows from the theorem that
$d$ is a divisor of $m$.

\begin{corollary}\label{cor:root2}
If an equation $D_t(G)=[F,G]$ admits a nontrivial solution in $\FF((T^{-1}))$
then it admits, as well, a solution of the form
\[
 G=f_mT^m+\dots+f_1T+g_0+g_{-1}T^{-1}+\ldots~\in\FF((T^{-1})).
\]
\end{corollary}
\begin{proof}
Let $G$ be a solution of degree $m$ which exists according to theorem
\ref{th:root1}. One can assume that the leading terms of $G$ and $F$ coincide,
without loss of generality. Let $\deg(F-G)=l$. If $l\le0$ then the statement is
true, let us consider the case $1\le l<m$.

Application of theorem \ref{th:root1} to equation $D_t(G)=[F-G,G]$ proves that
there exists a series $H\in\FF((T^{-1}))$ of degree $l$ which satisfies
equations $D_t(H)=[F-G,H]$. The series $H=h_lT^l+\dots$ commutes with $G$ and
$f_l-g_l=ch_l$. Therefore, the series $G'=G+cH$ satisfies the original equation
and $\deg(F-G')<l$. Repeating this arguments several times, if necessary, we
come to a solution which coincides with $F$ up to the term $f_1T$ inclusively.
\end{proof}

Notice, that the Lax equation may admit, in first $m$ orders of $T$, a solution
$G_{>0}$ which contains additional parameters comparing with $F_{>0}$, even if
$m$ is the minimal degree of the true solution $G$. In such a case, these
parameters vanish automatically in the process of solving equations for the
rest coefficients of $G$.

\begin{example}
Let us consider the lattice equation
\[
 u_{,t}=f[u]=u_1u^2u_{-1}(u_2-u_{-2})
\]
related via the substitution $u_1u=v$ with the modified Volterra model on the
stretched lattice $v_{,t}=v^2(v_2-v_{-2})$. This substitution acts on the
higher symmetries as well and this guarantees (see next section) the
solvability of the Lax equation with
\[
 F=f_*=u_1u^2u_{-1}T^2+u^2u_{-1}(u_2-u_{-2})T+\cdots.
\]
It is easy to check that the Lax equation admits, in first two orders,
solutions of degree 2 and 1:
\[
 G_{>0}=F_{>0}+cuu_{-1}T,\quad H_{>-1}=uu_{-1}T+u_1u_{-1}-uu_{-2},
\]
moreover, $G_{>0}=(H^2+cH)_{>0}$. However, the equation for the third
coefficient of $H$ does not admit a solution in $\FF$. For the solution $G$,
this means that the constant $c$ must be set equal to zero when computing the
fourth coefficient.
\end{example}

%-------------------------------------------------------------------------------
\section{Formal symmetry}\label{s:G}

Let us recall basic notions of the symmetry approach, in the context of scalar
evolutionary lattice equations
\begin{equation}\label{ut}
 \partial_t(u_n)=T^n(f(u_m,\dots,u_{\bar m})).
\end{equation}
For the sake of definiteness, we will write arguments of functions in
decreasing order of subscripts, moreover, we will assume (applying the
reflection $u_n\to u_{-n}$ if needed) that
\[
 f^{(m)}\ne0,\quad f^{(\bar m)}\ne0,\quad m\ge1,\quad m\ge\bar m
\]
where $f^{(j)}=\partial_j(f)$, $\partial_j=\partial/\partial u_j$. The numbers
$m$ and $-\bar m$ are called {\em order} and {\em lower order} of the lattice
equation. Any function $a\in\FF$ gives rise to the {\em evolutionary
derivative} $\nabla_a$ and the {\em linearization operator} $a_*$:
\[
 \nabla_a=\sum_{j\in\Integer}T^j(a)\partial_j,\qquad
 a_*=\sum_{j\in\Integer}a^{(j)}T^j\in\FF[T,T^{-1}].
\]
We will use also the notation $D_t=\nabla_f$ for the differentiation in virtue
of equation (\ref{ut}). The following identities are easy to prove:
\[
 [\nabla_a,T]=0,~~ (T(a))_*=Ta_*,~~
 \nabla_a(b)=b_*(a),~~ (\nabla_a(b))_*=\nabla_a(b_*)+b_*a_*.
\]
The lattice equation
\begin{equation}\label{utau}
 \partial_\tau(u_n)=T^n(g(u_k,\dots,u_{\bar k}))
\end{equation}
is called {\em symmetry} of equation (\ref{ut}) if the condition
\[
 \nabla_f(g)=\nabla_g(f)
\]
holds identically with respect to $u_j$. It means that the flows
$\partial_t,\partial_\tau$ commute (guaranteeing the existence of a common
generic solution $u_n(t,\tau)$). The lattice equation is considered integrable
if it admits symmetries of order arbitrarily large. The linearization of the
latter equation brings it to the operator form
\begin{equation}\label{fg}
 \nabla_f(g_*)=\nabla_g(f_*)+[f_*,g_*]
\end{equation}
which is more convenient for the analysis. Neglecting of the term
$\nabla_g(f_*)$ which is of a fixed degree in $T$ brings to equation
\begin{equation}\label{fG}
 D_t(G)=[f_*,G].
\end{equation}
Its solutions are called {\em formal symmetries} of the lattice equation
(\ref{ut}).

\begin{theorem}\label{th:GG}
If the lattice equation (\ref{ut}) admits symmetries (\ref{utau}) with $k$
arbitrarily large then equation (\ref{fG}) admits a solution
$G\in\FF((T^{-1}))$ of the form
\begin{equation}\label{G}
 G=f^{(m)}T^m+\dots+f^{(1)}T+g_0+g_{-1}T^{-1}+\cdots.
\end{equation}
If the lattice equation (\ref{ut}) with $\bar m<0$ admits symmetries
(\ref{utau}) with $-\bar k$ arbitrarily large then equation (\ref{fG}) admits a
solution $\bar G\in\FF((T))$ of the form
\begin{equation}\label{barG}
 \bar G=f^{(\bar m)}T^{\bar m}+\dots+f^{(-1)}T^{-1}+\bar g_0+\bar g_1T+\cdots.
\end{equation}
\end{theorem}
\begin{proof}
It is sufficient to prove the first statement, taking the change $u_n\to
u_{-n}$, $T\to T^{-1}$ into account. The series
$g_*=g^{(k)}T^k+\dots+g^{(1)}T+o(T)\in\FF((T^{-1}))$ satisfies equation
(\ref{Gt}) in the orders $T^{k+m},\dots,$ $T^{m+1}$. The procedure of the root
extraction described in theorem \ref{th:root1} yields, taking remark
\ref{rem:approx} into account, a series $G\in\FF((T^{-1}))$ such that
\begin{equation}\label{Gappro}
 D_t(G)=[f_*,G]+o(T^{2m-k+1}),\quad \deg G=m.
\end{equation}
Moreover, one can choose $G_{>0}=(f_*)_{>0}$ (if $k>m$) according to corollary
\ref{cor:root2}. Since $k$ is arbitrarily large, hence equation (\ref{Gt}) is
solvable in all orders of $T$.
\end{proof}

The conditions of solvability of equation (\ref{fG}) with respect to the
coefficients of the series $G$ or $\bar G$ serve as the necessary integrability
conditions. Equation (\ref{Gappro}) demonstrates that existence of a symmetry
of order $k\ge m+r$ implies that $r$ conditions are fulfilled, for the
coefficients $g_0,\dots,g_{-r+1}$. A symmetry with $k\le m$ gives no conditions
because it is `lost' on the background of the trivial symmetry with $g=f$ which
corresponds to the operator $G=(f_*)_{>0}$. Analogously, if there exist
symmetries with lower order $-\bar k=-\bar m+r$ then the solvability conditions
are fulfilled for the coefficients $\bar g_0,\dots,\bar g_{r-1}$ of the series
$\bar G$. Unfortunately, we do not know how many conditions must be checked in
order to guarantee the existence of at least one symmetry, even of small order.
Nevertheless, these conditions are rather convenient both for testing and
classification purposes, because we write them intermediately in terms of the
right hand side of the equation and their form does not depend on the actual
orders of symmetries which are not known in advance.

The first condition and a corollary of the $m$-th one are especially simple.
Let us make use of the following simple property:
\[
 \Res[A,B]\in\im(T-1),\quad \Res A:=\coef_{T^0}A,\quad A,B\in\FF((T^{-1}))
\]
(indeed, $[aT^j,bT^{-j}]=(T^j-1)(T^{-j}(a)b)$).

\begin{statement}\label{st:m0}
If the lattice equation (\ref{ut}) admits a symmetry (\ref{utau}) of order
$k>2m$ then there exist functions $\sigma,\sigma_1\in\FF$ such that
\begin{equation}\label{ss}
 D_t(\log f^{(m)})=(T^m-1)(\sigma),\quad
 D_t(f^{(0)}+\sigma)=(T-1)(\sigma_1).
\end{equation}
If (\ref{ut}) admits a symmetry (\ref{utau}) with $\bar k\le2\bar m<0$ then
there exist functions $\bar\sigma,\bar\sigma_1\in\FF$ such that
\begin{equation}\label{barss}
 D_t(\log f^{(\bar m)})=(T^{\bar m}-1)(\bar\sigma),\quad
 D_t(f^{(0)}+\bar\sigma)=(T-1)(\bar\sigma_1).
\end{equation}
\end{statement}
\begin{proof}
According to theorem \ref{th:GG}, equation (\ref{fG}) is solvable in the orders
of $T^{2m},\dots,T^0$, moreover, we can assume $G_{>0}=(f_*)_{>0}$. Application
of $\Res$ to this equation and its equivalent form
\[
 D_t(G)G^{-1}=f_*-G-G(f_*-G)G^{-1}
\]
yields, respectively,
\[
 D_t(g_0)\in\im(T-1),\quad D_t(f^{(m)})/f^{(m)}=(T^m-1)(g_0-f^{(0)})
\]
which is equivalent to (\ref{ss}); equations (\ref{barss}) are obtained in a
similar way.
\end{proof}

In general, equation (\ref{fG}) in each order of $T$ is of the form
\[
 A_j(g_j)=b_j,\quad A_j=T^m-\frac{T^j(a)}{a},\quad a=f^{(m)},\quad
 j=0,-1,-2,\dots
\]
where $b_j$ is a known expression which contains the coefficients of $f_*$ and
$g_0,\dots,g_{j+1}$. Thus, the integrability test for a given lattice equation
amounts to a step-by-step checking of whether $b_j\in\im A_j$; if not then the
equation is not integrable, if yes then we have to compute $g_j$ and to go to
the next condition. Notice, that equation $A_{j+m}(g)=b$ is reduced to
$A_j(\tilde g)=\tilde b$ under the change $g=T^j(a)\tilde g$, so that the
following problems appear: to characterize the image and to compute the
pre-image of the operators of the form
\[
 T^m-1,\quad T^m-\frac{T(a)}{a},~\dots,\quad T^m-\frac{T^{m-1}(a)}{a}
\]
with a given function $a$. The solution is well known at $m=1$:
\[
 \im(T-1)=\ker\Eu,\quad
 \Eu=\frac{\delta}{\delta u}=\sum_{j\in\Integer}T^{-j}\partial_j
\]
where $\Eu$ is called the {\em Euler operator} or the {\em variational
derivative}, while the pre-image of $T-1$ can be found by a difference version
of the integration by parts algorithm or by use of the homotopy operator, see
e.g.\,\cite{Kupershmidt_1985,Hydon_Mansfield_2004}. At $m>1$, the problem
admits a constructive solution as well, although it is more complicated (in
particular, the answer depends on whether $\log a$ belongs to the image of the
operator $T^{m-d}+\dots+T^d+1$ where $d|m$).

%-------------------------------------------------------------------------------
\section{Formal conservation law}\label{s:R}

The symmetric case $\bar m=-m$ is of most interest, because only such type of
lattice equations may admit higher order conservation laws. Let us recall that
a function $\rho\in\FF$ is called a {\em density of conservation law} for the
lattice equation (\ref{ut}) if there exists a function $\sigma\in\FF$ such that
\begin{equation}\label{frho}
 \nabla_f(\rho)=(T-1)(\sigma).
\end{equation}
A density is called {\em trivial} if $\rho\in\im(T-1)$ and two densities are
called equivalent if their difference is trivial. In order to factor out the
trivial conservation laws, let us apply the Euler operator to (\ref{frho}),
this yields the equation for $r=\Eu(\rho)=\rho_*^\dag(1)$:
\[
 \nabla_f(r)+f_*^\dag(r)=0
\]
where $\dag:\FF((T^{\pm1}))\to\FF((T^{\mp1}))$ denotes the conjugation
$(aT^j)^\dag=T^{-j}a$. Application of linearization once again yields the
operator equation
\begin{equation}\label{fr}
 \nabla_f(r_*)+f_*^\dag r_*+r_*f_*
  +\sum_{m\ge i,j\ge\bar m}T^{-j}(rf^{(i,j)})T^{i-j}=0.
\end{equation}
Notice, that the operator $r_*$ is symmetric, $r_*=r_*^\dag$. In particular,
$r$ depends on a symmetric set of variables: $r=r(u_k,\dots,u_{-k})$, $r^{(\pm
k)}\ne0$. The number $k$ is called the {\em order of the conservation law}. It
is easy to see that the degrees with respect to $T$ of the four terms in
equation (\ref{fr}) are equal to $k$, $k-\bar m$, $k+m$ and $M\le m-\bar m$,
respectively. This implies that equation (\ref{ut}) with $\bar m\ne-m$ can not
possess conservation laws of order $k>\min(m,-\bar m)$. So, we assume that
$\bar m=-m$ in what follows, that is, the lattice equation is of the
form
\begin{equation}\label{utm}
 \partial_t(u_n)=T^n(f(u_m,\dots,u_{-m})),\quad f^{(\pm m)}\ne0.
\end{equation}
The last sum in (\ref{fr}) is of a fixed degree with respect to $T$. Neglecting
it brings to equation
\begin{equation}\label{fR}
 D_t(R)+f_*^\dag R+Rf_*=0
\end{equation}
and its solution in the form of a series from $\FF((T^{-1}))$ or $\FF((T))$ is
called a {\em formal conservation law} for the lattice equation (\ref{utm}).
Equation (\ref{fR}) is invariant with respect to the conjugation, so we may
restrict ourselves with consideration of series from $\FF((T^{-1}))$.

\begin{theorem}\label{th:GR}
Let the lattice equation (\ref{utm}) admits conservation laws of order $k$
arbitrarily large. Then equation (\ref{fG}) admits solutions of the form
\begin{align}
\label{G'}
 G&=f^{(m)}T^m+\dots+f^{(1)}T+g_0+g_{-1}T^{-1}+\ldots~\in\FF((T^{-1})),\\
\label{barG'}
 \bar G&=f^{(-m)}T^{-m}+\dots+f^{(-1)}T^{-1}
   +\bar g_0+\bar g_1T+\ldots~\in\FF((T))
\end{align}
and equation (\ref{fR}) admits a solution of the form
\begin{equation}\label{R}
  R=r_lT^l+r_{l-1}T^{l-1}+\ldots~\in\FF((T^{-1})),\quad 0\le l<m
\end{equation}
such that
\begin{equation}\label{RG}
 \bar G^\dag R=-RG.
\end{equation}
\end{theorem}
\begin{proof}
Let us consider operators $r'_*$, $r_*$ corresponding to the conservation laws
of orders $k'>k>m$ as series from $\FF((T^{-1}))$. It follows from (\ref{fr})
that these series satisfy equation (\ref{fR}) in the first $k-m$ orders of $T$.
The extraction of the root from the series $r_*^{-1}r'_*$ of degree $k'-k>0$
brings to a series of degree $m$ which also satisfies equation (\ref{fG}) in
first $k-m$ orders, that is up to the terms $o(T^{3m-k+1})$. Since $k$ is
arbitrarily large, hence equation (\ref{fG}) is solvable for all orders of $T$
and we come to a solution of the form (\ref{G'}), taking the corollary
\ref{cor:root2} into account.

Conservation laws of orders $k=qm+l$, $0\le l<m$ constitute an infinite set at
least for one value of $l$. For the corresponding series $r_*$, the series
$r_*G^{-q}$ of degree $l$ satisfies equation (\ref{fR}) in first $(q-1)m+l$
orders, whence the existence of the solution $R$ (\ref{R}) follows.

The constructed series $G,R\in\FF((T^{-1}))$ allow to obtain the series $\bar
G=-(RGR^{-1})^\dag\in\FF((T))$ which satisfies equation (\ref{fG}). Moreover,
\[
 (\bar G^\dag)_{>0}=-RG_{>0}R^{-1}=-R(f_*)_{>0}R^{-1}=(f_*^\dag)_{>0}
\]
according to (\ref{fR}), therefore $\bar G_{<0}=(f_*)_{<0}$.
\end{proof}

It is clear from comparing with theorem \ref{th:GG} (at $\bar m=-m$) that
assumption about existence of an infinite set of conservation laws brings to
more restrictive integrability conditions than assumption about the higher
symmetries:
\[
 \begin{array}{ccc}
 \text{conservation laws} & & \text{higher symmetries} \\
 \Downarrow && \Downarrow\\
 G,R&\Rightarrow & G,\bar G
 \end{array}
\]
A weak point of the conditions which follow from the equation for $R$ is that
the degree $l=\deg R$ is not known in advance, so we have to inspect the values
$l=0,\dots,m-1$. In particular, we come to the following statement instead of
\ref{st:m0}.

\begin{statement}\label{st:m0'}
If the lattice equation (\ref{utm}) admits two conservation laws of orders
$k'>k>3m$ then there exist functions $\sigma,\sigma_1,s,s_1\in\FF$ and an
integer $l$, $0\le l<m$ such that
\begin{gather}
\label{ss'}
 D_t(\log f^{(m)})=(T^m-1)(\sigma),\quad D_t(f^{(0)}+\sigma)=(T-1)(\sigma_1),\\
\label{sss}
 \log(-T^l(f^{(m)})/f^{(-m)})=(T^m-1)(s),\quad D_t(s)+2f^{(0)}=(T-1)(s_1).
\end{gather}
\end{statement}
\begin{proof}
In the notations of theorem \ref{th:GR}, the series $G,\bar G,R$ constructed
from $r'_*,r_*$ satisfy equations (\ref{fG}), (\ref{fR}) in first $2m+1$
orders. Equations (\ref{ss'}) are proved like in statement \ref{st:m0}. First
equation (\ref{sss}) follows from (\ref{fR}) in the leading order, with the
function $s=\log(f^{(-m)}r_l)$. Multiplication of (\ref{fR}) by $R^{-1}$ and
applying $\Res$ results in $D_t(\log r_l)+2f_0\in\im(T-1)$ which is equivalent
to the second equation (\ref{sss}).
\end{proof}

%-------------------------------------------------------------------------------
\section{The lattices of order 1}\label{s:1}

The integrability conditions simplify drastically for the first order lattice
equations
\begin{equation}\label{ut1}
 \partial_t(u_n)=T^n(f(u_1,u,u_{-1})),\quad f^{(\pm1)}\ne0.
\end{equation}
In this case, the Lax equation for $G=f^{(1)}T+g_0+g_{-1}T^{-1}+\dots$ turns
out to be equivalent to a sequence of conservation laws (possibly trivial)
defined by certain recurrent relations. In order to write them down we will use
the polynomials
\begin{gather*}
 P_0=1,\quad P_1(x_1)=x_1,\quad P_2(x_1,x_2)=x_2+\frac{x^2_1}{2},\\
 P_3(x_1,x_2,x_3)=x_3+x_1x_2+\frac{x^3_1}{6},~ \dots
\end{gather*}
defined by the generating function
\[
 P_0[x]+P_1[x]\lambda+P_2[x]\lambda^2+\ldots
  =\exp(x_1\lambda+x_2\lambda^2+x_3\lambda^3+\ldots).
\]
These polynomials are well known in the representation theory of
infinite-dimensional Lie algebras (see e.g. \cite{Jimbo_Miwa_1983,
Kac_Raina_1987}) and are related to the complete exponential Bell polynomials
$Y_k$ \cite{Comtet_1974} by the change
$k!P_k(x_1,\dots,x_k)=Y_k(x_1,\dots,k!x_k)$.

In order to rewrite the equation $D_t(G)=[f_*,G]$ in the form of conservation
laws, we make use of the fact that it serves as the compatibility condition for
equations
\begin{equation}\label{psi}
 G(\psi)=\lambda\psi,\quad D_t(\psi)=f_*(\psi).
\end{equation}
Let us consider expansions with respect to $\lambda$ of the ratios
\begin{equation}\label{pq}
 p=T(\psi)/\psi,\quad q=D_t(\psi)/\psi,
\end{equation}
this brings to equations
\begin{gather}
\label{gp}
 f^{(1)}p+g_0-\lambda+\frac{g_{-1}}{T^{-1}(p)}
  +\frac{g_{-2}}{T^{-1}(p)T^{-2}(p)}+\ldots=0,\\
\label{q}
 q=f^{(1)}p+f^{(0)}+\frac{f^{(-1)}}{T^{-1}(p)},\\
\label{pt}
 D_t(p)/p=(T-1)(q).
\end{gather}
It is easy to see that equation (\ref{gp}) defines an invertible change between
the coefficients of the series $G$ and
$p=p_{-1}\lambda+p_0+p_1\lambda^{-1}+\dots$:
\[
 p_{-1}=\frac{1}{f^{(1)}},\quad p_0=-\frac{g_0}{f^{(1)}},\quad
 p_1=-\frac{g_{-1}T^{-1}(f^{(1)})}{f^{(1)}},~~\dots\,,
\]
so that $G\in\FF((T^{-1}))$ if and only if $p\in\FF((\lambda^{-1}))$. Moreover,
(\ref{q}) implies that
$q=\lambda-\sigma_0-\sigma_1\lambda^{-1}-\ldots\in\FF((\lambda^{-1}))$ and a
solution of equations (\ref{psi}) (a formal Baker--Akhiezer function) is
constructed by integration of equations (\ref{pq}) as a series of the form
\[
 \psi(n)=a(n)\lambda^n(1+a_1(n)\lambda^{-1}+a_2(n)\lambda^{-2}+\dots)
\]
with coefficients in a certain extension of the field $\FF$. One more
change
\begin{equation}\label{prho}
 p=\frac{\lambda}{f^{(1)}}\exp(-\rho_1\lambda^{-1}-\rho_2\lambda^{-2}-\dots)
\end{equation}
brings equations (\ref{q}) and (\ref{pt}) to relations
\begin{gather}
\label{1.rs}
 D_t(\rho_j)=(T-1)(\sigma_j),\quad j\ge0,\\
\label{1.r01}
 \rho_0=\log f^{(1)},\quad \rho_1=f^{(0)}+\sigma_0,\\
\label{1.rj}
 P_{j+1}[-\rho]+f^{(-1)}T^{-1}(f^{(1)}P_{j-1}[\rho])+\sigma_j=0,\quad j>0.
\end{gather}
The existence of a formal symmetry $G\in\FF((T^{-1}))$ is equivalent to
solvability of equations (\ref{1.rs}) with respect to $\sigma_j\in\FF$,
moreover, the densities $\rho_{j+1}$ are explicitly found from (\ref{1.rj}):
\begin{align*}
 \rho_2&=f_{-1}T^{-1}(f_1)+\frac{1}{2}\rho_1^2+\sigma_1,\\
 \rho_3&=f_{-1}T^{-1}(f_1\rho_1)+\rho_1\rho_2-\frac{1}{6}\rho_1^3+\sigma_2,\\
 \rho_4&=f_{-1}T^{-1}(f_1(\rho_2+\frac{1}{2}\rho_1^2))
    +\rho_1\rho_3+\frac{1}{2}\rho^2_2-\frac{1}{2}\rho_1^2\rho_2
    +\frac{1}{24}\rho_1^4+\sigma_3,~\dots\,.
\end{align*}

In order to write analogously the consequences from existence of the second
formal symmetry $\bar G\in\FF((T))$, we consider a function $\bar\psi$ which
satisfies equations
\[
 \bar G^\dag(\bar\psi)=-\lambda\bar\psi,\quad
 D_t(\bar\psi)=-f_*^\dag(\bar\psi).
\]
It is easy to check, as before, that the ratios
\[
 \bar p=T(\bar\psi)/\bar\psi,\quad \bar q=D_t(\bar\psi)/\bar\psi
\]
can be expanded into series with respect to $\lambda$, of the form
\[
 \bar p=-\frac{\lambda}{T(f^{(-1)})}
  \exp(-\bar\rho_1\lambda^{-1}-\bar\rho_2\lambda^{-2}-\dots),\quad
 \bar q=\lambda-\bar\sigma_0-\bar\sigma_1\lambda^{-1}-\dots
\]
with coefficients $\bar\rho_j,\bar\sigma_j\in\FF$, and that these ratios
satisfy equations
\[
 \bar q=-T(f^{(-1)})\bar p-f^{(0)}
  -\frac{T^{-1}(f^{(1)})}{T^{-1}(\bar p)},\quad
  D_t(\bar p)/\bar p=(T-1)(\bar q).
\]
This implies the same recurrent relations (\ref{1.rs}), (\ref{1.rj}) for
functions $\bar\rho_j,\bar\sigma_j$ as for $\rho_j,\sigma_j$, but with the
initial data
\begin{equation}\label{1.r01'}
 \bar\rho_0=\log(-T(f^{(-1)})),\quad \bar\rho_1=-f^{(0)}+\bar\sigma_0
\end{equation}
instead of (\ref{1.r01}). If lattice equation (\ref{ut1}) admits, in addition,
a formal conservation law $R=r_0+r_{-1}T^{-1}+\ldots\in\FF((T^{-1}))$ then it
follows from equations (\ref{fR}), (\ref{RG}) that one can take
$\bar\psi=R(\psi)$. Then the series
\[
 s=\log(\bar\psi/\psi)=s_0+s_1\lambda^{-1}+\ldots\in\FF((\lambda^{-1}))
\]
satisfies the equations
\[
 \log\bar p-\log p=(T-1)(s),\quad D_t(s)=\bar q-q
\]
that is, both density sequences are related by equations
\begin{equation}\label{1.rr}
 \rho_j-\bar\rho_j=(T-1)(s_j),\quad j\ge0.
\end{equation}
Here, we may exclude the functions $\bar\sigma_j=\sigma_j-D_t(s_j)$ from
consideration, because $\bar\rho_j$ can be found from the recurrent
relations
\begin{gather*}
 \bar\rho_0=\log(-T(f^{(-1)})),\quad \rho_1-\bar\rho_1=2f^{(0)}+D_t(s_0),\\
 P_{j+1}[-\bar\rho]-P_{j+1}[-\rho]\mspace{350mu} \\ \mspace{50mu}
  +f^{(-1)}T^{-1}(f^{(1)}(P_{j-1}[\bar\rho]-P_{j-1}[\rho]))=D_t(s_j),
 \quad j>0.
\end{gather*}

It was already mentioned in the Introduction that the classification of
integrable lattice equations (\ref{ut1}) requires, according to
\cite{Yamilov_1983}, only three simplest conditions (1 of the type (\ref{1.rs})
and 2 of the type (\ref{1.rr}), cf. also with statement \ref{st:m0'}) which can
be cast into the form
\begin{gather*}
 D_t(\log f^{(1)})\in\im(T-1),\\
 \log(-f^{(1)}/f^{(-1)})=(T-1)(s),\quad D_t(s)+2f^{(0)}\in\im(T-1).
\end{gather*}
These conditions can be derived under assumptions that the equation admits a
symmetry of order $k\ge2$ and a conservation law of order $k'\ge3$, or that it
admits a pair of conservation laws of orders $k'>k\ge3$. The analysis of these
conditions is a rather tedious task which brings to a finite list of equations.
One can prove by inspection that all of them admit an infinite set of higher
symmetries (and conservation laws, except for several degenerate cases like the
linear equation $\partial_t(u_n)=u_{n+1}-u_{n-1}$ which does not admit
conservation laws of order $>0$, but equation (\ref{fR}) admits the solution
$R=1$).

For the sake of completeness, let us prove the statement on the equivalence of
the constructed conservation laws (\ref{1.rs}) with the standard definition
through the residues of the powers of formal symmetry. The proof follows,
essentially, to Kupershmidt \cite[ch.\,IX.3]{Kupershmidt_1985} and it is an
adaptation for the discrete case of the proof by Flaschka \cite{Flaschka_1983}.

\begin{statement}
Let the lattice equation (\ref{ut1}) admits the formal symmetry
$G=f^{(1)}T+\ldots\in\FF((T^{-1}))$ and the quantities $\rho_j$ are defined by
relations (\ref{gp}), (\ref{prho}), then
\begin{equation}\label{ResG}
 \Res G^j-j\rho_j\in\im(T-1),\quad j=1,2,\dots\,.
\end{equation}
\end{statement}
\begin{proof}
Equations (\ref{psi}), (\ref{pq}) for the $\psi$-function imply
\[
 T(\psi)=p\psi=\sum p_j\lambda^{-j}\psi=\sum p_jG^{-j}(\psi)
\]
from where the identity
\begin{equation}\label{pg}
 T=p_{-1}G+p_0+p_1G^{-1}+p_2G^{-2}+\ldots
\end{equation}
follows which is equivalent to relation (\ref{gp}) between the series $G$ and
$p$. Let us introduce the notations
\[
 G_j=(G^j)_{\ge0},\quad A_j=(G^j)_{<0},\quad \epsilon_j=\coef_{T^{-1}}G^j
\]
(in particular, $\epsilon_{-1}=1/T^{-1}(f^{(1)})$, $\epsilon_0=0$). Right
multiplication of (\ref{pg}) by $G^j$ and neglecting of the negative powers of
$T$ result in
\[
 TG_j+T(\epsilon_j)
  =p_{-1}G_{j+1}+p_0G_j+\dots+p_jG_0,\quad j=-1,0,1,\dots
\]
which is equivalent to equation (a difference version of the Cherednik formula
\cite{Cherednik_1978})
\begin{equation}\label{T-p}
 (T-p){\cal G}=-T(E)
\end{equation}
for the generating series
\[
 {\cal G}=1+\lambda^{-1}G_1+\lambda^{-2}G_2+\dots,\quad
 E=\epsilon_{-1}\lambda+\epsilon_0+\epsilon_1\lambda^{-1}+\dots\,.
\]
Application of $\Res$ yields
\begin{equation}\label{ResGG}
 \Res{\cal G}=1+\lambda^{-1}\Res G+\lambda^{-2}\Res G^2+\ldots
  =\frac{T(E)}{p}.
\end{equation}
Next, let us denote
\[
 \alpha_j=A_j(\psi)/\psi=\frac{\epsilon_j}{T^{-1}(p)}
  +\frac{a_{j,-2}}{T^{-1}(p)T^{-2}(p)}+\ldots
  =\epsilon_jT^{-1}(f^{(1)})\lambda^{-1}+\dots
\]
then
\[
 G_j(\psi)/\psi=\lambda^j-\alpha_j,\quad
 TG_j(\psi)/\psi=p(\lambda^j-T(\alpha_j)).
\]
Let us apply the identity (\ref{T-p}) to $\psi(\mu)$ and divide the result by
$\psi(\mu)$, this yields
\begin{align*}
 -T(E(\lambda))&=(T-p(\lambda)){\cal G(\lambda)}(\psi(\mu))/\psi(\mu)\\
 &=\sum_{j\ge0}\lambda^{-j}\Bigl(
  p(\mu)(\mu^j-T(\alpha_j(\mu)))-p(\lambda)(\mu^j-\alpha_j(\mu))\Bigr)\\
 &=\frac{\lambda(p(\mu)-p(\lambda))}{\lambda-\mu}
   +\sum_{j\ge0}\lambda^{-j}
    \Bigl(p(\lambda)\alpha_j(\mu)-p(\mu)T(\alpha_j(\mu))\Bigr).
\end{align*}
Division by $-p(\lambda)$ and passage to the limit $\mu\to\lambda$ brings to
equation
\[
 \frac{T(E)}{p}=\lambda\partial_\lambda(\log p)
  +\sum_{j\ge0}\lambda^{-j}(T-1)(\alpha_j)
\]
and the statement follows from comparison with (\ref{prho}) and (\ref{ResGG}).
\end{proof}

In conclusion, notice that an analog of the substitution $p=T(\psi)/\psi$ can
be defined for the lattice equations of order $m$ as well: one can prove that
an equation $G(\psi)=\lambda\psi$ with the series $G\in\FF((T^{-1}))$ of order
$m$ is equivalent to equation
\[
 (T^m-p_{m-1}T^{m-1}-\dots-p_0)(\psi)=0,\quad p_j\in\FF((\lambda^{-1})).
\]
In the matrix notations, equations (\ref{psi}) are replaced by
\[
 T(\Psi)=P\Psi,\quad D_t(\Psi)=Q\Psi
\]
where
\[
 \Psi=\begin{pmatrix}
  \psi\\[0.5ex] \vdots \\[1ex] T^{m-1}(\psi)
   \end{pmatrix},\qquad
 P=\begin{pmatrix}
   0 & 1 &       &  \\
     &   &\ddots &  \\
     &   &       & 1\\
  p_0&   &\dots  & p_{m-1}
 \end{pmatrix}
\]
and elements of the matrix $Q$ are expressed through $f^{(j)}$ and $p_j$.
Equation (\ref{pt}) is replaced with
\[
 D_t(P)=T(Q)P-PQ
\]
which implies, in particular, that $\log p_0$ serves as a generating series for
the densities of conservation laws. However, in the context of the
integrability conditions this equation leads to rather cumbersome relations,
and, apparently, it gives no advantages comparing with the straightforward
solving of the Lax equation for the formal symmetry.

%-------------------------------------------------------------------------------
\paragraph{Acknowledgements.}
The discussion of this subject with R.I. Yamilov was of great importance for
me. Conditions (\ref{ss'}), (\ref{sss}), in a slightly weaker form, were
obtained by him in 2008, under assumption about existence of a formal symmetry
of arbitrary degree as a starting point \cite{Yamilov}.

Research for this article was supported by the RFBR grant 13-01-00402a and the
project SFB/TR 109 ``Discretization in Geometry and Dynamics''.

%-------------------------------------------------------------------------------

\end{document}